\documentclass[sigconf]{acmart}
\usepackage{booktabs}
\usepackage{balance}
\usepackage{amsmath}
\usepackage{amssymb}
\usepackage[tight,footnotesize]{subfigure}
\usepackage{graphicx}
\usepackage{bm}
\usepackage{algorithm}
\usepackage[noend]{algorithmic}
\usepackage{multirow}
\usepackage{diagbox}
\usepackage{caption}
\usepackage{amsfonts}
\usepackage{url}
\usepackage{breakurl}

\renewcommand{\paragraph}[1]{\smallskip \noindent {\textsc{#1}}}

\DeclareGraphicsExtensions{.eps,.ps,.pdf}
\makeatletter
\def\@copyrightspace{\relax}
\makeatother

\begin{document}
\title{Stochastic Coupon Probing in Social Networks}
{\author{Shaojie Tang}
\affiliation{%
  \institution{University of Texas at Dallas}}

\begin{abstract}
In this paper, we study stochastic coupon probing problem in social networks.  Assume there is a social network and  a set of  coupons.  We can offer coupons to some users \emph{adaptively} and those users who accept the offer will act as \emph{seeds} and influence their friends in the social network.  There are two constraints which are called the \emph{inner} and \emph{outer} constraints, respectively. The set of coupons redeemed by users  must satisfy inner constraints, and the set of all probed users must satisfy outer constraints. One seeks to develop a coupon probing policy that achieves the maximum influence while satisfying both inner and outer constraints. Our main result is a constant approximation policy for the stochastic coupon probing problem for any monotone submodular utility function. 
\end{abstract}

%
%
%


\maketitle

\section{Introduction}
\label{sec:intro}
Social media marketing has been recognized as  one of the most effective marketing methods which can increase marketers' business' visibility with little cost. As reported in Social Media Examiner, around 96\% of marketers are currently participating in social media marketing. Different from traditional online marketing strategy whose focus
is to find the best match between a given product and an individual customer, social media marketing is more concerned about the network value of a customer. For example, giving some influential customers an incentive (such as coupon as studied in this work) to share your product information with their social circle could increase your brand recognition rapidly as they are engaging with a broad audience of customers.

To this end, we introduce and study the stochastic coupon probing problem in social networks. Assume there is a social network and  a set of  coupons.  We can offer coupons to some users and those users who accept the offer will act as \emph{seed nodes} and influence their friends in the social network.  However, any feasible solution must satisfy two constraints which are called the \emph{inner} and \emph{outer} constraints, respectively. The set of coupons redeemed by users  must satisfy inner constraints, and the set of all probed users must satisfy outer constraints. In particular, we consider the following constraints in our basic model:
(Outer Constraint) The number of probes involving the same user can not exceed a non-negative constant, and (Inner Constraint) the total value of redeemed coupons can not exceed a non-negative number. Our ultimate goal is to develop a coupon probing policy that achieves the maximum influence subject to both inner and outer constraints.

We notice that most of existing studies on coupon allocation assume a ``budgeted and non-adaptive '' setting \cite{yangcontinuous}, where there
is a predefined budget for selecting seed nodes, their solution is to commit the set of initial users and corresponding coupons all at once in advance. 
Our model is largely different from all existing work. We consider a stochastic optimization problem: Rather than make a one-shot decision at the very beginning, we can probe users one by one using some carefully selected coupons, and the decision made at each round is depending on the realization from previously probed users.
 In particular, if a user accepts our coupon, then we add that user to the solution and deduct corresponding coupon value from our budget, otherwise, the budget remains the same. Since the remaining budget after each round depends on the realization from earlier rounds, our decision made in each round is also dependent on the choices in earlier rounds.

Based on the above discussion, our problem falls belong to the category of stochastic optimization. Although there has been extensive work on adaptive/stochatic submodular optimization  \cite{golovin2011adaptive}\cite{badanidiyuru2016locally}\cite{adamczyk2016submodular}, our stochastic coupon probing model is substantially different from these problems: first of all, existing studies mainly assume that the cost of the action is fixed and pre-known, however, this assumption is clearly not true under our model, i.e, the actual cost of offering some coupon to a user depends on whether or not she accepts the offer; secondly, actions may incur non-uniform cost in our model, i.e., coupons have different values; thirdly, the realization of different actions are not independent, i.e., any rational user will not accept a low value coupon while rejecting a high value coupon; lastly,  our model involves a broader set of constraints including outer matroid constraint and inner knapsack constraint. Unfortunately, none of the existing solutions in the literature explicitly take the above three constraints into account. 

To the best of our knowledge, we are the first to systematically study the stochastic coupon probing problem in social networks. We present a novel probing policy that achieves a constant approximation ratio. This research contributes fundamentally to the development of approximate probing policies for any problems that fall into the family of stochastic optimization problems with correlated realizations subject to  outer matroid constraint and inner knapsack constraint. 

The rest of this paper is organized as follows: We review
related work in Section \ref{sec:review} and introduce our problem formulation in Section \ref{sec:pre}. We develop a novel coupon probing policy in Section \ref{sec:adaptive}. In Section \ref{sec:ext}, we extend the basic model to incorporate more constraints. 
We conclude this paper in Section \ref{sec:conclusion}. Some important notations are listed in Table \ref{symbol}.

\begin{table}[t]\centering
\caption{Symbol table.}
\begin{tabular}{{|c|l|}}
\hline
\textbf{Notation} & \textbf{Meaning}\\
\hline
\hline
$\mathcal{C}$ & the ground set that contains all coupons\\
\hline
$\mathcal{C}_l$ & a list of all low-value coupons\\
\hline
$I(\mathcal{U})$ & the expected cascade of $\mathcal{U}$\\
\hline
$\mathbf{y}$ & a decision matrix\\
\hline
$f^+(\mathbf{y})$ & the concave extension of $f$\\
\hline
$p_{v,c}$ & the attractiveness of $c$ to $v$\\
\hline
$(v,c)$ & an action that offers $c$ to $v$\\
\hline
$(v,\psi)$ &  an action that offers $\psi$ to $v$ sequentially\\
\hline
$\phi((v,c))$ & the realization of action $(v, c)$ \\
\hline
$b_{v\psi}$ & the expected cost of action $(v, \psi)$\\
\hline
$\Psi$ & subsequences from $\mathcal{C}_l$ with size at most $K$ \\
\hline
$\mathcal{S} \triangleq \mathcal{V}\times \Psi$ & the action space under \textbf{P.B}\\
\hline
\end{tabular}
\label{symbol}
\end{table}

\section{Related Work}
\label{sec:review}
This work is closely related to two topics: influence maximization (IM) and stochastic submodular optimization. IM has been extensively studied in the literature \cite{kempe2003maximizing,chen2013information,leskovec2007cost,cohen2014sketch,chalermsook2015social}, their objective is to find a set of influential customers so as to maximize the expected cascade. However, our work differ from all existing studies in several major aspects. Traditional IM assumes any node is guaranteed to be activated once it is selected, we relax this assumption by allowing users to response differently to different coupon values.  Recently, \cite{yangcontinuous}\cite{yuan2017adaptive} study discount allocation problem in social networks. However, \cite{yangcontinuous} assumes a budgeted and non-adaptive setting where decisions must be made all at once in advance, while  \cite{yuan2017adaptive} simply ignores the outer constraint, e.g., their model allows the company to probe all users at zero cost, we formulate our problem as a stochastic optimization problem subject to both inner and other constraints.

Another closely related topic to ours is adaptive/stochatic submodular optimization  \cite{golovin2011adaptive}\cite{seeman2013adaptive}\cite{badanidiyuru2016locally}\cite{adamczyk2016submodular}. Existing approaches can not apply directly to our setting due to the following reasons:  (1) the actual cost of offering some coupon to a user is stochastic rather than a fixed value, (2) coupons have different values, (3)  the realization of different actions are not independent, and (4) our problem involves a broader set of constraints including outer matroid constraint and inner knapsack constraint. In summary, we are the first to study the stochastic coupon probing problem with correlated realizations subject to outer matroid constraint and inner knapsack constraint. We propose a novel probing policy that provides the first bounded approximate solutions to this problem. 

\section{Preliminaries and Problem Formulation}
\label{sec:prel}
\subsection{Submodular Function and Its Concave Extension}
\label{sec:pre}
Consider a set function $h:2^{\Omega}\rightarrow \mathbb{R}$, where $2^{\Omega } $ denotes the power set of  $\Omega$. 
We say $h$ is submodular if and only if it satisfies he follow property: For every $X, Y \subseteq \Omega$ with $X \subseteq Y$ and every $x \in \Omega \backslash Y$, we have that $h(X\cup \{x\})-h(X)\geq h(Y\cup \{x\})-h(Y)$.

Given any vector $\bold{x}=\{x_1,x_2,\ldots,x_n\}$ such that each $0\leq x_i\leq 1$. 
The concave extension of $h$ is defined as
 \begin{equation*}
h^+(\mathbf{x})=\max \left\{\sum_{X\subseteq \Omega} \alpha_X f(X) \middle\vert \begin{array}{l}
    \alpha_X \geq 0, \sum_{X\subseteq \Omega}\alpha_X \leq1;\\
       \sum_{X}\alpha_X\mathbf{1}_X \leq \mathbf{x}
  \end{array}\right\}\end{equation*}

\subsection{Models}
\label{sec:models}
\subsubsection{Coupon Adoption Model and Influence Function:} Consider a set of users $\mathcal{V}$ and a set of  coupon values $\mathcal{C}$, for every coupon $c\in \mathcal{C}$  and every user $v \in \mathcal{V}$, we define the attractiveness of  $c$ to $v$ as $p_{v,c}\in [0,1]$. For simplicity of presentation, we directly use $c$ to represent the value of $c$. We assume that users are rational, meaning that they favor coupons with larger value, e.g.,  $\forall v\in \mathcal{V}: p_{v,c}\leq p_{v,c'}$ if $c\leq c'$.

We next describe the coupon adoption model used in this work. At the beginning, every user $v\in \mathcal{V}$ selects a random number $\sigma_v \in [0,1]$, which is called \emph{threshold} of $v$. And $v$ accepts any coupon with attractiveness larger than or equal to $\sigma_v$.  Once a user accepts a coupon, she becomes a seed node of our product and starts to influence other users in the social network. We denote the expected influence of a set of seed nodes $\mathcal{U}$ as $I(\mathcal{U})$, which is a monotone and submodular function of $\mathcal{U}$, following the seminal work on influence maximization \cite{kempe2003maximizing}. It is easy to verify that $p_{v,c}$ is equivalent to the adoption probability as defined in \cite{yangcontinuous}, e.g., the probability that $v$ accepts $c$ is $p_{v,c}$. 

\subsubsection{Action Space and Realizations}  We consider an action space $\mathcal{H} \triangleq \mathcal{V}\times \mathcal{C}$, where  an \emph{action} $(v, c) \in \mathcal{H} $ represents offering coupon $c\in \mathcal{C}$ to user $v\in \mathcal{V}$. Define $\phi((v,c))\in\{0,1\}$ to be the realization of action $( v, c)\in \mathcal{H}$:

\[\phi((v,c)) = \begin{cases}
1 & \mbox{if $v$ accepts $c$  and becomes the seed ($p_{v,c}\geq \sigma_v$)}\\
0 & \mbox{if $v$ rejectes $c$ ($p_{v,c}< \sigma_v$)}
\end{cases}\]
We say $(v,c)$ is \emph{active} (resp. \emph{inactive}) if $\phi((v,c))=1$ (resp. $\phi((v,c))=0$). We can observe the realization of $\phi((v, c))$ by probing $v$ using $c$. According to our coupon adoption model,  each action $(v, c)$ is active with probability $p_{v,c}$. However, the states of different actions are not independent, e.g., if $(v, c)$ is active, then every $(v, c')$ with $c'\geq c$ must be active (any rational user will not reject a coupon with higher value). 

\subsubsection{Probing Policy} Following the framework of \cite{golovin2011adaptive}, we characterize a probing policy $\pi$  as a function from realizations observed so far to $\mathcal{H}$, specifying which action to pick next based on our current knowledge.
 \[\pi: \Phi\rightarrow \mathcal{H},\]
 where $\Phi$ denotes the collection of realizations observed so far. By following a particular policy, we sequentially pick an action $(v, c)$ and observes its state, then decides the next action to pick, and so on. If $v$ accepts $c$ ($\phi((v,c)) =1$), we add $v$ to the set of seeds and pay $c$, otherwise ($\phi((v,c)) =0$), we pay nothing and move to the next round.

Given a probing policy $\pi$, let $\Pr[\mathcal{U}|\pi]$ denote the probability that $\mathcal{U}$ becomes the set of seed nodes by following $\pi$, then the expected utility of $\pi$ can be represented as
\[f(\pi)= \sum_{\mathcal{U}\subseteq \mathcal{V}} \Pr[\mathcal{U}|\pi] I(\mathcal{U})\]

\subsection{Problem Formulation}
We consider the setting in which we are given constraints on both users probed and the coupons redeemed by all probed users. In particular, the set of coupons redeemed by users  must satisfy inner constraint, and the set of all probed users must satisfy outer constraint. In the basic model, we consider the following two constraints.
\begin{itemize}
\item \underline{Outer Constraints:} The number of probes involving the same user can not exceed a non-negative constant  $K$. This constraint models the fact that  pushing too many promotions  to the same user \cite{tang2016optimizing} could impede her user experience.
\item \underline{Inner Constraints:} The total value of redeemed coupons can not exceed a non-negative number $B$.
\end{itemize}
We present the formal definition of our problem in \textbf{P.A}. Our ultimate goal is to identify the best probing policy subject to both inner and outer constraints.
 \begin{center}
\framebox[0.4\textwidth][c]{
\enspace
\begin{minipage}[t]{0.4\textwidth}
\small
$\textbf{P.A}$ $\max_{\pi} f(\pi)$\\
\textbf{subject to:}
\emph{inner constraint} and \emph{outer constraint};
\end{minipage}
}
\end{center}
\vspace{0.1in}
We develop a constant approximation coupon probing policy in Section \ref{sec:adaptive}. Later in Section \ref{sec:ext}, we extend the basic model to incorporate more constraints.

\textbf{ A Toy Example.} For better understanding of our model, we next go through a toy example. In this example, we assume $\mathcal{C}=\{1, 2\}$; $\mathcal{V}=\{a,b,c,d,e\}$; $K=1$, e.g., the same user can be probed at most once; $B=3$. Consider a given probing policy $\pi$,

 In the first round, assume $\pi(\emptyset)=(d, 1)$, we offer coupon $1$ to user $d$, e.g., pick action $(d, 1)$, and observe that  $d$ rejects $1$. The remaining budget is $2$ and the realization is  $\Phi_1=\{\phi(d, 1)=0\}$.

 In the second round, assume $\pi(\Phi_1)=(a, 2)$, we offer $2$ to $a$, e.g., pick action $(a, 2)$, and observe that $a$ accepts $2$. The remaining budget is $1$ and the realization is $\Phi_2=\{\phi(d, 1)=0, \phi(a, 2)=1\}$.

In the third round,  assume $\pi(\Phi_2)=(b, 1)$, we offer $1$ to $b$, e.g., pick action $(b, 1)$, and observe that $b$ accepts $1$. The remaining budget is $0$ and the realization is $\Phi_3=\{\phi(d, 1)=0, \phi(a, 2)=1, \phi(b, 1)=1\}$.

At the end, $\mathcal{U}=\{a, b\}$ is returned as the seed set and the utility is $I(\{a,b\})$.

\section{Stochastic Coupon Probing Policy}
\label{sec:adaptive}
\textbf{Overview.} In this paragraph, we present an overview of our algorithm design.  We first construct two candidate algorithms, say \verb"ALG I" and \verb"ALG II", and then randomly pick one algorithm with equal probability as the final solution. Our approach is inspired by the enumeration trick used to solve knapsack problem: we partition all coupons into two groups according to their values, then apply   \verb"ALG I" (resp. \verb"ALG II") to find a near optimal solution using only ``low-value'' (resp. ``high-value'') coupons. Later we show that the better solution returned from the above two algorithms achieves at least a constant approximation ratio.

\subsection{ALG I}
Before introducing the design of \verb"ALG I", we first present a restricted version of $\mathbf{P.A}$.
Notice that in $\mathbf{P.A}$, there is no restriction on the sequence of probed users and coupons, meaning that we can probe one user first and then come back to probe that user again, using a different coupon. To facilitate  our study, we next introduce a restricted version of $\mathbf{P.A}$ by adding one additional outer constraint: the same user can only be probed in consecutive rounds. The restricted outer constraints is presented as follows:
\begin{itemize}
\item \underline{Restricted Outer Constraints:} (a) the number of probes involving the same user can not exceed a non-negative constant $K$; (b) the same user can only be probed in consecutive rounds.
\end{itemize}

We present the formal definition of the restricted problem in \textbf{P.B}.

 \begin{center}
\framebox[0.48\textwidth][c]{
\enspace
\begin{minipage}[t]{0.45\textwidth}
\small
$\textbf{P.B}$ $\max_{\pi} f(\pi)$\\
\textbf{subject to:}
\emph{restricted outer constraints} and \emph{inner constraints};
\end{minipage}
}
\end{center}
\vspace{0.1in}

Perhaps surprisingly, later we show that restricting ourselves to probe the same user in consecutive rounds does not sacrifice much in terms of the utility. This enables us to focus on solving  $\textbf{P.B}$ in the rest of this paper.

Since we can only probe the same user in consecutive rounds, a valid probing policy on any individual user $v$ can be characterized as a sorted list of coupons  $\psi\subseteq \mathcal{C}$ with size at most $K$, specifying the sequence of coupons we offer to $v$ sequentially. Assume we decide to probe $v$ using $\psi$, the probing process that involves $v$ can be roughly described as follows: We follow  $\psi$ to offer a coupon $c\in \psi$ to $v$ one by one. If $v$ accepts $c$, we deduct $c$ from the budget and stop. Otherwise, if $v$ turns down $c$, we simply move to the next coupon in $\psi$, the budget of the next round remains unchanged. This process iterates until either $v$ accepts the current coupon or $v$ rejects the entire $\psi$.

Let $\mathcal{C}_l=\{c\in \mathcal{C}| c\leq B/2\}$ be the low-value coupons in $\mathcal{C}$. The input of \verb"ALG I" is $\mathcal{C}_l$, restricting ourselves to use low-value coupons only. Notice that there is no point offering a larger coupon to a user before offering her a smaller coupon.  Therefore, we sort $\mathcal{C}_l$  in non-decreasing order of their values, e.g., $\forall i\leq j, c_i\leq c_j$.
Define \[\Psi\triangleq \{\psi| \psi \prec \mathcal{C}_l, |\psi|\leq K\},\] where $\psi \prec \mathcal{C}_l$ represents that $\psi$ is a subsequence of $\mathcal{C}_l$. Since $K$ is a constant, the size of $\Psi$ is  polynomial in the size of $\mathcal{C}_l$. We next define the action space $\mathcal{S}$ under $\textbf{P.B}$: \[\mathcal{S} \triangleq \mathcal{V}\times \Psi\] Picking an action  $(v,\psi) \in \mathcal{S} $ translates to probing user $v$ using sequence $\psi$. 

Now we are ready to describe \verb"ALG I". 
The general idea is to first find a fractional solution with a bounded approximation ratio and then round it to an integral solution.

We present \textbf{P.B-r}, a relaxed version of \textbf{P.B}, as follows.
 \begin{center}
\framebox[0.45\textwidth][c]{
\enspace
\begin{minipage}[t]{0.45\textwidth}
\small
\textbf{P.B-r:}
\emph{Maximize $f^+(\mathbf{y})$}\\
\textbf{subject to:}
\begin{equation*}
\begin{cases}
\forall v\in \mathcal{V}: \sum_{\psi\in \Psi} y_{v\psi}\leq1 \quad(C1.1)\\
  \sum_{(v,\psi)\in \mathcal{S}} y_{v\psi}b_{v\psi}\leq B \quad(C2.1)\\
\forall (v,\psi)\in \mathcal{S}: y_{v\psi} \in[0,1] \quad (C3.1)
\end{cases}
\end{equation*}
\end{minipage}
}
\end{center}
\vspace{0.1in}
In the above formulation,
\begin{itemize}
\item $\mathbf{y}$ is a $|\mathcal{V}|\times |\Psi|$  decision matrix.
\item $f^+(\mathbf{y})$ is a concave extension of $f$:
\begin{equation}
\max \left\{\sum_{(v,\psi)\in \mathcal{S}} \alpha_{(v,\psi)} f((v,\psi)) \middle\vert \begin{array}{l}
    \alpha_{(v,\psi)} \geq 0;\\
      \sum_{(v,\psi)\in \mathcal{S}}\alpha_{(v,\psi)}\leq1;\\
      \forall v: \sum_{(v,\psi) \ni [v\psi] }\alpha_{(v,\psi)} \leq y_{v\psi}
  \end{array}\right\}
  \label{eq:1}
\end{equation}
\item $b_{v\psi}$ is the expected cost of action $(v, \psi)$:
\[b_{v\psi}=\sum_{c\in \psi}(\prod_{c': c' \leq c} (1-p_{u, c'} )p_{v, c} c )\]
\end{itemize}
\begin{algorithm}[h]
{\small
\caption{Continuous Greedy}
\label{alg:greedy-peak}
\begin{algorithmic}[1]
\STATE Set $\delta=1/(|\mathcal{V}|\times |\Psi|)^2, t=0, f(\emptyset)=0$.
\WHILE{$t<1$}
\STATE Let $R(t)$ be a random set which contains each $(v,\psi)$ independent with probability $y_{v\psi}(t)$.
\STATE For each $(v,\psi)\in \mathcal{S}$, estimate
\[\omega_{v\psi}=\mathbb{E}[f(\mathcal{R}(t)\cup\{(v,\psi)\})]-\mathbb{E}[[f(\mathcal{R}(t))]\]
\STATE Solve the following liner programming problem and obtain the optimal solution $\mathbf{y}^*$
\STATE
\begin{center}
\framebox[0.42\textwidth][c]{
\enspace
\begin{minipage}[t]{0.42\textwidth}
\small
\textbf{P.C:}
\emph{Maximize $\sum_{[v\psi]\in \mathcal{S}}\omega_{v\psi}y_{v\psi}$ }\\
\textbf{subject to:}\begin{equation*}
\begin{cases}
(C1.1) \mbox{ and } (C3.1) \\
  \sum_{(v,\psi)\in \mathcal{S}} y_{v\psi}b_{v\psi}\leq \beta B
\end{cases}
\end{equation*}
\end{minipage}
}
\end{center}
\vspace{0.1in}
\STATE Let $y^g_{v\psi}(t+\delta)=y^g_{v\psi}(t)+\delta{y}^*_{v\psi}$; \label{line:1}
\STATE Increment $t=t+\delta$;
\ENDWHILE
\end{algorithmic}
}
\end{algorithm}

We use the continuous greedy algorithm (Algorithm \ref{alg:greedy-peak}) to solve \textbf{P.B-r} and obtain an fractional solution $\mathbf{y}^{g}$. This algorithm is first proposed in \cite{vondrak2008optimal} in the context of offline submodular
function maximization. Notice that  condition (2.1) is replaced by  $\sum_{(v,\psi)\in \mathcal{S}} y_{v\psi}b_{v\psi}\leq \beta B$ in \textbf{P.C}, where $0\leq\beta\leq1/2$ is a tuning parameter, e.g., the budget $K$ is scaled down by a factor of $\beta$.

Given $\mathbf{y}^g$, we next introduce the design of \verb"ALG I".

\emph{Step 1:}  Let $R$ be a random set obtained by including each
element $(v, \psi)$ independently with probability ${y}^g_{v\psi}$.

\emph{Step 2:} Apply Contention resolution \cite{vondrak2011submodular} to $\mathcal{R}$ to obtain a new set $\mathcal{Z}$ such that $|\mathcal{Z}\cap \mathcal{S}^v|\leq1$ where $\mathcal{S}^v$ contains all actions from $\mathcal{S}$ that involves $v$.

\emph{Step 3:} Consider actions from $\mathcal{Z}$ one by one in an arbitrary order, let the current action, say $(v, \psi)$, survive if the current budget is no smaller than $B/2$, else we discard it.
 \begin{itemize}
 \item If $(v,\psi)$ survives, probe $v$ using $\psi$.
 \begin{itemize}
 \item If $v$ rejects the entire $\psi$, repeat Step 3 with the same budget,
 \item Otherwise, if $v$ accepts some coupon $c\in \psi$, deduct $c$ from the current budget and repeat Step 3.
 \end{itemize}
 \item If $(v,\psi)$ is discarded, repeat Step 3 with the same budget.
\end{itemize}
\emph{Step 4:} We add all survived actions to $T_1$.

Let $\mathbf{y}^+$ be the optimal solution to \textbf{P.B-r}, we next prove that \verb"ALG I" is feasible and its expected cascade $\mathbb{E}[f(T_1)]$ is close to $f^+(\mathbf{y}^+)$.

 \begin{lemma}
 \label{lemma:2}
\verb"ALG I" is feasible and \[\mathbb{E}[f(T_{1})]\geq (1-1/e)(1-\beta)(1-2\beta)\beta  f^+(\mathbf{y}^+)\]
 \end{lemma}
\begin{proof}We first prove the feasibility of \verb"ALG I". First of all, we always choose some action from $\mathcal{Z}$ to probe, thus every user receives at most $K$ coupons (outer constraint is satisfied). In Step 3, we probe a user if and only if the current budget is no smaller than $B/2$, the inner constraint is also satisfied since we only consider low-value coupons.

Fix some $(v,\psi)\in \mathcal{S}$. After Step 2, $(v,\psi)$ survives with probability $(1-\beta){y}^g_{v\psi}$. Now we consider the survival rate of $(v,\psi)$ after the Step 3. Let $e_1$ denote the event that the total value of the coupons redeemed by $\mathcal{V}\setminus \{v\}$  is no larger than $B/2$. Since every coupon in $\psi$ has value no larger than $B/2$, $(v, \psi)$ will be probed if $e_1$ happens. Because the expected cost of our solution is bounded by $\sum_{(v,\psi)\in \mathcal{S}} y_{v\psi}b_{v\psi}\leq \beta B$, then according to Markov's inequality, the probability of $e_1$ is bounded by $\Pr[e_1] \geq 1-2\beta$. It follows that  \begin{align*}
\mathbb{E}[f(T_{1})]&\geq (1-\beta)(1-2\beta) f(\mathbf{y}^g) \\
&\geq (1-1/e)(1-\beta)(1-2\beta)\beta  f^+(\mathbf{y}^+)
 \end{align*}
 The first inequality is proved in \cite{vondrak2011submodular} and the second inequality is due to  $f(\mathbf{y}^g)\geq (1-1/e)\beta f^+(\mathbf{y}^+)$ which is proved in \cite{calinescu2011maximizing}.
\end{proof}

\subsection{ALG II} In \verb"ALG II", we only use the largest coupon $c_{\max}=\arg\max_{c\in \mathcal{C}} c$. Given that we can only use one type of coupon, the coupon probing problem is reduced to determining the best order of users in which to probe them. To avoid trivial cases, we assume $c_{\max}>B/2$, otherwise \verb"ALG I" already achieves a constant approximation ratio.

The design of \verb"ALG II" follows a simple greedy manner. We sort all users in non-increasing order of $I(\{v\})$, then probe them one by one using $c_{\max}$. This process terminates when the current user accepts $c_{\max}$ or the last user rejects $c_{\max}$. Let $\mathcal{U}_v$ denote all users placed before $v$, we have 
\begin{lemma}
\label{lem:qian}
\verb"ALG II" is the optimal policy if we are only allowed to use the largest coupon $c_{\max}$, and the expected utility of \verb"ALG II" is $\sum_{v\in \mathcal{V}}(\prod_{u: u \in \mathcal{U}_v} (1-p_{u, c_{\max}} )p_{v, c_{\max}} I(v))$.
\end{lemma}

\subsection{Put It All Together} We randomly choose one algorithm from \verb"ALG I" and \verb"ALG II" with equal probability as our solution. We refer to this algorithm as Stochastic Coupon Probing policy (\textsc{stoch-CP}). The rest of this paper is devoted to proving  Theorem \ref{thm:main1}.
\begin{theorem}
\label{thm:main1}
Given any $0\leq\beta\leq1/2$, \textsc{stoch-CP} achieves \[\frac{(1-1/e)(1-\beta)(1-2\beta)\beta}{2}\] approximation ratio.
\end{theorem}

\begin{proof}
Let $\textbf{P.A}_l$ be a low-value version of \textbf{P.A}, restricting ourselves to using low-value coupons $\mathcal{C}_l$ only.
Let $\pi^*$ be the optimal solution to \textbf{P.A}. We next show by construction that there exists a policy $\pi^\diamond$ that is feasible to $\textbf{P.A}_l$  and its performance is close to $\pi^*$. Given $\pi^*$, we next explain how to build $\pi^\diamond$ on $\pi^*$. Note that our purpose is \textbf{not} to really construct such $\pi^\diamond$ (this turns out to be impossible since $\pi^*$ is not known in advance), instead, we only want to prove the existence of such a policy.

Let $\Phi_t$ be the set of realizations available at round $t$ (starting with $\Phi_0=\emptyset$). Assume $\pi^*$ is known and $\pi^*(\Phi_t)=(v^*, c^*)$, e.g., $\pi^*$  picks $(v^*, c^*)$ as the next action conditioned on $\Phi_t$, the new policy $\pi^\diamond(\Phi_t)$ is defined as follows:
\[\pi^\diamond(\Phi_t) = \begin{cases}
(v^*, c^*) & \mbox{if $c^* \leq B/2$}\\
\emptyset & \mbox{otherwise}
\end{cases}\]

After observing the realization of  $\phi(\pi^\diamond(\Phi_t))$, we update the realizations from $\Phi_{t}$  to $\Phi_{t+1}$ as follows:
\begin{itemize}
\item \underline{Case A:} If $c^* \leq B/2$, add $\phi((v^*, c^*))$ to  $\Phi_{t}$: $\Phi_{t+1}=\Phi_{t}\cup \phi((v^*, c^*))$.
 \item \underline{Case B:} If $c^* > B/2$, then $\pi^\diamond$ is not allowed to use $c^*$, thus it is impossible to observe the value of $\phi((v^*, c^*))$. To overcome this difficulty, we guess the value of $\phi((v^*, c^*))$ using simulation as follows:

   Assume $c<1/2$ is the largest low-value coupon offered to $v^*$ by $\pi^\diamond$ so far, randomly select a threshold $\sigma_v \in (p_{v^*, c}, 1]$. This value is fixed once it has been selected and it will be used as $v^*$'s threshold in the subsequent rounds.
 \begin{itemize}
  \item If $p_{v^*, c^*}\geq \sigma_v $, add $\phi((v^*, c^*))=1$ to $\Phi_{t+1}$ virtually
  \item If $p_{v^*, c^*}< \sigma_v $, we add $\phi((v^*, c^*))=0$ to $\Phi_{t+1}$ virtually.
   \end{itemize}
 \end{itemize}
We use $\Phi_f$ to denote the realizations obtained at the final round. We define $\mathcal{U}_{\Phi_f}$ to be the seed set under $\Phi_f$. Note that $\pi^\diamond$ mimics the execution of $\pi^*$, thus
\begin{equation}E[I(\mathcal{U}_{\Phi_f})|\pi^\diamond]=f(\pi^*)
\label{eq:1}
 \end{equation}
 However, $I(\mathcal{U}_{\Phi_f})$ is not the actual set of seeds obtained by $\pi^\diamond$, this is because under Case B, we may \emph{virtually} add some $\phi((v',c'))=1$ with $c' > B/2$ to $\Phi_f$. Fortunately, we add at most one such action to $\Phi_f$, otherwise it violates the budget constraint. It follows that the actual seed set obtained by $\pi^\diamond$ is $\mathcal{U}_{\Phi_f}\setminus \{v'\}$ if such $v'$ exists, thus the expected utility of $\pi^\diamond$ is $f(\pi^\diamond)=E[I(\mathcal{U}_{\Phi_f}\setminus \{v'\})|\pi^*]$. Then we have
\begin{equation}
\label{eq:2}f(\pi^\diamond)=E[I(\mathcal{U}_{\Phi_f})|\pi^*]-E[ \Delta_{\mathcal{U}_{\Phi_f}}(v')|\pi^*],
\end{equation}
 where $\Delta_{\mathcal{U}_{\Phi_f}}(v')=\mathcal{U}_{\Phi_f}-\mathcal{U}_{\Phi_f}\setminus\{v'\}$ is marginal utility of $v'$ given existing seed set $\mathcal{U}_{\Phi_f}\setminus\{v'\}$.

Equations (\ref{eq:1}) and (\ref{eq:2}) imply that
\begin{equation}
f(\pi^*)=f(\pi^\diamond)+E[ \Delta_{\mathcal{U}_{\Phi_f}}(v')|\pi^*]
 \end{equation}

To prove Theorem \ref{thm:main1}, it is sufficient to show that $E[f(T_1)]\geq (1-1/e)(1-\beta)(1-2\beta)\beta f(\pi^\diamond)$ (Lemma \ref{lem:1}) and $E[f(T_2)]\geq E[ \Delta_{\mathcal{U}_{\Phi_f}}(v')|\pi^*]$ (Lemma \ref{lem:2}).

\begin{lemma} \label{lem:1} $E[f(T_1)]\geq (1-1/e)(1-\beta)(1-2\beta)\beta  f(\pi^\diamond)$.
\end{lemma}
\begin{proof}
Let $\pi^*_l$ be the optimal solution to $\textbf{P.A}_l$.
Fix some $v\in \mathcal{V}$. Consider a realization $\phi_{\overline{v}}$ that involves $\mathcal{V}\setminus \{v\}$, we define $d$ as the longest possible sequence of coupons offered to $v$ by $\pi^*$ (not necessarily to be in consecutive rounds) given  $\phi_{\overline{v}}$, e.g., this can only happen when the entire $d$ is rejected by $v$. It is important to note that when  $\pi^*$ and  $\phi_{\overline{v}}$ are given, $d$ is also fixed, and we call  $d$ the probing sequence picked by $\pi^*_l$ under $\phi_{\overline{v}}$. Let $\Phi^d_{\overline{v}}$ denote the set of all realizations that involves $\mathcal{V}\setminus \{v\}$ under which  $d$ is offered to $v$ by $\pi^*$. For every $[v\psi]\in \mathcal{S}$, we define $y^*_{v\psi}$ as the probability that $d$ is offered to $v$ by $\pi^*_l$, e.g., $y^*_{v\psi}=\sum_{\phi_{\overline{v}}\in \Phi^d_{\overline{v}}} \Pr[\phi_{\overline{v}}]$. We first prove that $\mathbf{y}^*_l$ is a feasible solution to \textbf{P.B}.
Because $\pi^*_l$ is a feasible policy, the following conditions are satisfied in each round
\begin{enumerate}
\item Every user is probed using at most $K$ coupons.

\item The total value of redeemed coupons under  $\pi^*_l$ is no larger than $B$. Because $\sum_{[v\psi]\in \mathcal{S}} y^*_{v\psi}b_{v\psi}$ is the expected cost of $\pi^*_l$, we have $\sum_{[v\psi]\in \mathcal{S}} y^*_{v\psi}b_{v\psi} \leq B$.
\end{enumerate}
The above two properties ensure that $\mathbf{y}^*_l$ is a feasible solution to \textbf{P.B-r}. Since  $\mathbf{y}^+$ is the optimal solution to \textbf{P.B-r}, we have $f^+(\mathbf{y}^+)\geq f^+(\mathbf{y}^*_l) \geq f(\pi^*_l)$. Because $\pi^*_l$ is the optimal solution to $\textbf{P.A}_l$, we have $f(\pi^*_l)\geq f(\pi^\diamond)$. It follows that $f^+(\mathbf{y}^+)\geq f(\pi^\diamond)$. Then we have
\begin{align*}
\mathbb{E}[f(T_{1})]&\geq (1-1/e)(1-2\beta)\beta  f^+(\mathbf{y}^+)\\
&\geq (1-1/e)(1-2\beta)\beta  f(\pi^\diamond)
\end{align*}
The first inequality is due to  Lemma \ref{lemma:2} and the second inequality is due to $f^+(\mathbf{y}^+)\geq f(\pi^\diamond)$.
\end{proof}

\begin{lemma}
\label{lem:2}
$E[f(T_2)]\geq E[ \Delta_{\mathcal{U}_{\Phi_f}}(v')|\pi^*]$.
\end{lemma}
\begin{proof} Because $f$ is submodular, we have $E[ \Delta_{\mathcal{U}_{\Phi_f}}(v')|\pi^*]\leq E[I(\{v'\})|\pi^*]$. According to \verb"ALG II", $T_2$ contains  a single active user which gives the largest expected influence, then we have $E[f(T_2)]\geq E[ \Delta_{\mathcal{U}_{\Phi_f}}(v')|\pi^*]$.
\end{proof}
This finishes the proof of Theorem \ref{thm:main1}.
\end{proof}
\section{Extension: Incorporating Coupon Distribution Budget}
\label{sec:ext}
We now study \textbf{P.A1}, an extension of $\textbf{P.A}$,  by adding one more constraint to $\textbf{P.A}$: the number of users that are probed can not exceed an non-negative number $W$. This constraint captures the fact that the company often has limited budgeted on coupon
producing and distribution, thus we can offer coupons to a limited number of users. We next summarize the updated constraints as follows.
\begin{itemize}
\item \underline{Outer Constraints:} The number of probes involving the same user can not exceed a non-negative number $K$ and the number of users probed can not exceed an non-negative number $W$.
\item \underline{Inner Constraints:} The total value of redeemed coupons can not exceed a non-negative number $B$.
\end{itemize}

\subsection{Extended ALG I} We present \textbf{P.B-r1}, an extended version of \textbf{P.B-r}, as follows. Condition (C4.1) specifies the additional outer constraint.
 \begin{center}
\framebox[0.45\textwidth][c]{
\enspace
\begin{minipage}[t]{0.45\textwidth}
\small
\textbf{P.B-r1:}
\emph{Maximize $f^+(\mathbf{y})$}\\
\textbf{subject to:}
\begin{equation*}
\begin{cases}
\forall v\in \mathcal{V}: \sum_{d\in \Psi} y_{v\psi}\leq1 \quad(C1.1)\\
  \sum_{(v,\psi)\in \mathcal{S}} y_{v\psi}b_{v\psi}\leq B \quad(C2.1)\\
\forall (v,\psi)\in \mathcal{S}: y_{v\psi} \in[0,1] \quad (C3.1)\\
\sum_{v\in \mathcal{V}, \psi \in \Psi} y_{v\psi}\leq W \quad(C4.1)\\
\end{cases}
\end{equation*}
\end{minipage}
}
\end{center}
\vspace{0.1in}

Similar to \verb"ALG I", the basic idea of  \verb"Extended ALG I" (\verb"E-ALG I") is still to use the continuous greedy algorithm to obtain a fractional solution first, then round it to an integer solution. The continuous greedy algorithm is basically the same as Algorithm \ref{alg:greedy-peak}, except in that \textbf{P.C} is replaced by \textbf{P.C1}.

\begin{center}
\framebox[0.42\textwidth][c]{
\enspace
\begin{minipage}[t]{0.42\textwidth}
\small
\textbf{P.C1:}
\emph{Maximize $\sum_{[v\psi]\in \mathcal{S}}\omega_{v\psi}y_{v\psi}$ }\\
\textbf{subject to:}\begin{equation*}
\begin{cases}
(C1.1) \mbox{ and } (C3.1) \\
  \sum_{(v,\psi)\in \mathcal{S}} y_{v\psi}b_{v\psi}\leq \beta B\\
  \sum_{v\in \mathcal{V}, \psi\in \Psi} y_{v\psi}\leq \beta W
\end{cases}
\end{equation*}
\end{minipage}
}
\end{center}

Let $\mathbf{y}^g$ be the solution returned from the continuous greedy algorithm, we next round $\mathbf{y}^g$ to an integer solution. Our rounding approach is the same as the one used in \verb"ALG I" except in that Step 2.0 is replaced by the following procedure.

\vspace{0.2in}
\emph{Modified Step 2:} Apply Contention resolution \cite{vondrak2011submodular} for two matroids to $R$ to obtain a new set $\mathcal{Z}$ such that $|\mathcal{Z}\cap \mathcal{S}^v|\leq1$ and $|\mathcal{Z}|\leq W$.
\vspace{0.2in}

Let $T_2$ be the set of survived actions, we next prove that \verb"E-ALG I" is feasible and its expected cascade $\mathbb{E}[f(T_2)]$ is close to the optimal solution.

 \begin{lemma}
 \label{lemma:3}
\verb"E-ALG I" is feasible and \[\mathbb{E}[f(T_{2})]\geq (1-1/e)(1-\beta)^2(1-2\beta)\beta f^+(\mathbf{y}^+)\]
 \end{lemma}
\begin{proof}We first prove the feasibility of \verb"E-ALG I". First of all, because $|\mathcal{Z}\cap \mathcal{S}^v|\leq1$ and $|\mathcal{Z}|\leq W$, $\mathcal{Z}$ satisfies the outer constraints. In Step 3, we probe a group of users subject to budget constraint, thus inner constraint is also satisfied.

Fix some $(v,\psi)\in \mathcal{V}\times \Psi$. After Step 2, $(v,\psi)$ survives with probability $(1-\beta)^2 y^g_{v\psi}$. Similar to Lemma \ref{lemma:2}, we can prove that  the survival rate of $(v,\psi)$ after Step 3 is $1-2\beta$. 
Moreover, the above rounding process is equivalent to performing three independent rounding subject to two outer constraints and one inner constraints respectively, and returning the intersection of the three sets. Then according to \citep{vondrak2011submodular}, we have $\mathbb{E}[f(T_{1})]\geq (1-\beta)^2(1-2\beta)\beta f(\mathbf{y}^g)$. It follows that  $\mathbb{E}[f(T_{1})]\geq (1-1/e)(1-\beta)^2(1-2\beta)\beta  f^+(\mathbf{y}^+)$.
\end{proof}

\subsection{Extended ALG II}
Similar to ALG II, we only use $c_{\max}$ in  \verb"Extended ALG II" (\verb"E-ALG II").   However, instead of using a simple greedy algorithm  to find the optimal solution, we turn to dynamic programming to handle the additional constraint  (C4.1). Notice that since we can only probe $W$ users, our problem becomes a joint user selection and sequencing problem, e.g., select a set of $W$ users and decide in what order to probe them using $c_{\max}$.

Assume all users are sorted in non-decreasing order of $I(\{v\})$. In the following recurrence, $f[v_i,l]$ stores the optimum expected utility that can be obtained from users $v_i, \cdots , v_n$ subject to the constraint that we can only probe $l$ users. Essentially, we need to decide whether to probe $v_i$ or not in each step. The optimal value gained from adding $v_i$ to the solution can be calculated as \[I(\{v_i\})  p_{v_i,c_{\max}} + (1-p_{v_i,c_{\max}}) f[i-1,l-1]\] Otherwise, the optimal value is $f[i-1, l]$. Thus we adopt the following recurrence to obtain a sequence of users $f[n,W]$:
\[f[i,l] = \max \{ I(\{v_i\}) p_{v_i,c_{\max}} + (1-p_{v_i,c_{\max}})) f[i-1,t-1], f[a_{i-1}, l]\} \]
Then we probe $f[v_n, W]$ one by one in non-increasing order of $I(\{v\})$ until some user accepts $c_{\max}$ or none of them accepts $c_{\max}$. Let $T_2$ denote the set of seeds returned from \verb"E-ALG II".
\subsection{Put It All Together} We randomly choose one algorithm from \verb"E-ALG I" and \verb"E-ALG II" with equal probability as our solution. We refer to this algorithm as Extended Stochastic Coupon Probing policy (\textsc{E-stoch-CP}). The rest of this paper is devoted to proving  Theorem \ref{thm:main2}.
\begin{theorem}
\label{thm:main2}
Given any $0\leq\beta\leq1/2$, \textsc{E-stoch-CP} achieves \[\frac{(1-1/e)(1-\beta)^2(1-2\beta)\beta}{2}\] approximation ratio.
\end{theorem}

\begin{proof} Let $\textbf{E-P.A}_l$ be a low-value version of \textbf{E-P.A}, restricting ourselves to using low-value coupons $\mathcal{C}_l$ only.
Let $\pi^*$ be the optimal solution to \textbf{E-P.A}.
Similar to the proof of Theorem \ref{thm:main1}, we can show that there exists a  policy $\pi^\diamond$ that is feasible to $\textbf{E-P.A}_l$ and
\begin{equation}
f(\pi^*)=f(\pi^\diamond)+E[ \Delta_{\mathcal{U}_{\Phi_f}}(v')|\pi^*]
 \end{equation}
 To prove Theorem \ref{thm:main2}, it is sufficient to prove the following two lemmas.
 \begin{lemma} \label{lem:4} $\mathbb{E}[f(T_{1})]\geq (1-1/e)(1-\beta)^2(1-2\beta)\beta  f(\pi^\diamond)$
\end{lemma}
The proof of Lemma \ref{lem:4} is analogous to the proof of Lemma \ref{lem:1}, thus omitted here.
\begin{lemma}
$E[f(T_2)]\geq E[ \Delta_{\mathcal{U}_{\Phi_f}}(v')|\pi^*]$.
\end{lemma}
\begin{proof}Since  $v'$ denotes the user (if any) who accepts some coupon larger than $B/2$ and we can add at most one such user to the optimal solution, one can obtain a valid upper bound of  $E[ \Delta_{\mathcal{U}_{\Phi_f}}(v')|\pi^*]$, e.g., the marginal contribution of $v'$ to the optimal policy, by computing $f(\pi^h)$ where $\pi^h$ is the optimal solution when restricting ourselves to using high-value coupons $\mathcal{C}\setminus \mathcal{C}_l$. It was worth noticing that those users who accept a low-value coupon under $\pi^*$ will not contribute to $E[ \Delta_{\mathcal{U}_{\Phi_f}}(v')|\pi^*]$. However, if a user, say $v$, rejects all low-value coupons offered by $\pi^*$,  we need to update the distribution of $\sigma_v$ as follows: assume $\hat{c}$ is the largest low-value coupon rejected by $v$, then the conditional value of $\sigma_v$  is to be randomly selected from $(p_{v, \hat{c}},1]$. Under either of the previous two cases, the conditional value of $E[ \Delta_{\mathcal{U}_{\Phi_f}}(v')|\pi^*]$ is only decreased. This observation indicates that $f(\pi^h)$ is still a valid upper bound of $E[ \Delta_{\mathcal{U}_{\Phi_f}}(v')|\pi^*]$, conditioned on some low-value coupons have been offered to the users.

 Next we prove that $E[f(T_2)]\geq E(\pi^h)$. We first prove that there exists a $\pi^h$ that probes the same user in consecutive rounds. Consider a relaxed problem where the set of users $\mathcal{W}$ probed by $\pi^h$ is known, it is optimal to probe $\mathcal{W}$ in non-increasing order of $I(\{v\})$. This is because if there exists two users $u,v\in \mathcal{W}: I(\{v\})< I(\{u\})$ and $v$ is probed first, we lose the chance to probe $u$ if $v$ accepts the coupon, thus probing $u$ first leads to a larger expected utility. Therefore, it is optimal to probe the same user in consecutive rounds. It follows that the expected utility of $\pi^h$ is

 \[\mathbb{E}_{\mathcal{W}\sim \pi^h}[\sum_{v\in \mathcal{W}} \prod_{u\in \mathcal{W}_v} (1-p_{u,c^h_u})p_{v,c^h_v}I(\{v\})],\]

where $\mathcal{W}_v$ denotes the set of users that are probed before $v$, e.g., $\forall u\in \mathcal{W}_v: I(\{u\})\leq I(\{v\})$, and $c^h_v$ denotes the largest coupon offered to $v$ by $\pi^h$. Consider a fixed $\mathcal{W}$, we next show that $\sum_{v\in \mathcal{W}} \prod_{u\in \mathcal{U}_v} (1-p_{u,c^h_u})p_{v,c^h_v}I(\{v\})$ is a increasing function of $p_{z, c^h_z}$ for any $z\in \mathcal{W}$ (assuming $p_{u,c^h_u}$ is fixed for all other $u\in \mathcal{U}\setminus \{z\}$). This is because $\sum_{v\in \mathcal{W}} \prod_{u\in \mathcal{W}_v} (1-p_{u,c^h_u})p_{v,c^h_v}I(\{v\})$ can be rewritten as

 {\small \begin{equation*}\hspace{-1in}
 \sum_{v\in \mathcal{W}_z} \prod_{u\in \mathcal{U}_v} (1-p_{u,c^h_u})p_{v,c^h_v}I(\{v\}) \end{equation*}
 \begin{equation} \label{eq:3} + \prod_{u\in \mathcal{U}_z} (1-p_{u,c^h_u}) (p_{v, c^h_z} \theta_1+ (1-p_{z, c^h_z})\theta_2)
  \end{equation}}

  In the above equation, $\theta_1$  is the expected utility conditioned on $z$ is probed and accepts the coupon and $\theta_2$ is the expected utility conditioned on $z$ is probed and rejects the coupon. Since the first term of Eq. (\ref{eq:3}) does not contain $p_{z, c^h_z}$, we next focus on proving that the second term of Eq. (\ref{eq:3}) is a increasing function of $p_{z, c^h_z}$. Because $\theta_1=I(\{z\})$ and all users probed after $z$ has expected utility no larger than $I(\{z\})$, we have $\theta_1\geq\theta_2$. It follows that the second term of Eq. (\ref{eq:3}) is an increasing function of $p_{z, c^h_z}$, thus \[\sum_{v\in \mathcal{W}} \prod_{u\in \mathcal{W}_v} (1-p_{u,c^h_u})p_{v,c^h_v}I(\{v\})\] is an increasing function of $p_{z, c^h_z}$. Therefore, offering $c_{\max}$ to each user will never decrease the expected utility.  Since $f(T_2)$ is the optimal solution when using $c_{\max}$ only, we have $E[f(T_2)]\geq E(\pi^h)$. \end{proof}
  This finishes the proof of Theorem \ref{thm:main2}.
  \end{proof}

\section{Conclusion}
\label{sec:conclusion}
In this work, we study stochastic coupon probing problem in social networks.  Different from existing work on discount allocation, we formulate our problem as a stochastic optimization problem. Our ultimate goal is to develop an efficient probing policy subject to inner and outer constraints. We notice that existing stochastic optimization approaches can not apply to our setting directly. Our main result is a constant approximation policy for the stochastic coupon probing problem, and we believe that this result made fundamental contributions to the field of stochastic optimization. 
\bibliographystyle{named}
\bibliography{reference}

\end{document}